\documentclass{amsart}
\usepackage{hyperref}
\usepackage{enumerate}

\newcommand*{\mailto}[1]{\href{mailto:#1}{\nolinkurl{#1}}}

\newtheorem{theorem}{Theorem}[section]
\newtheorem{lemma}[theorem]{Lemma}

\newcommand{\R}{\mathbb{R}}
\newcommand{\Z}{\mathbb{Z}}

\newcommand{\C}{\mathbb{C}}

\newcommand{\M}{\mathbb{M}}

\newcommand{\nn}{\nonumber}
\newcommand{\beq}{\begin{equation}}
\newcommand{\eeq}{\end{equation}}
\newcommand{\bea}{\begin{eqnarray}}
\newcommand{\eea}{\end{eqnarray}}

\newcommand{\ol}{\overline}
\newcommand{\pa}{\partial}

\newcommand{\I}{\mathrm{i}}
\newcommand{\E}{\mathrm{e}}

\DeclareMathOperator{\res}{Res}
 \newcommand{\noprint}[1]{}

\newcommand{\la}{\lambda}

\numberwithin{equation}{section}

\begin{document}
\title[On the form of dispersive shock waves of the KdV equation]{On the form of dispersive shock waves of the Korteweg--de Vries equation}

\author[I.\ Egorova]{Iryna Egorova}
\address{B. Verkin Institute for Low Temperature Physics\\ 47, Lenin ave\\ 61103 Kharkiv\\ Ukraine}
\email{\href{mailto:gladkazoya@gmail.com}{gladkazoya@gmail.com}}

\author[Z.\ Gladka]{Zoya Gladka}
\address{B. Verkin Institute for Low Temperature Physics\\ 47, Lenin ave\\ 61103 Kharkiv\\ Ukraine}
\email{\href{mailto:iraegorova@gmail.com}{iraegorova@gmail.com}}

\author[G.\ Teschl]{Gerald Teschl}
\address{Faculty of Mathematics\\ University of Vienna\\
Oskar-Morgenstern-Platz 1\\ 1090 Wien\\ Austria\\ and International Erwin Schr\"odinger
Institute for Mathematical Physics\\ Boltzmanngasse 9\\ 1090 Wien\\ Austria}
\email{\href{mailto:Gerald.Teschl@univie.ac.at}{Gerald.Teschl@univie.ac.at}}
\urladdr{\href{http://www.mat.univie.ac.at/~gerald/}{http://www.mat.univie.ac.at/\string~gerald/}}


\thanks{Zh. Mat. Fiz. Anal. Geom. {\bf 12}, 3--16 (2016)}
\thanks{Research supported by the Austrian Science Fund (FWF) under Grant V120.}

\begin{abstract}
We show that the long-time behavior of solutions to the Korteweg--de Vries shock problem can be described as a slowly modulated one-gap solution in the dispersive shock region.
The modulus of the elliptic function (i.e., the spectrum of the underlying Schr\"odinger operator) depends only on the size of the step of the initial data and
on the direction, $\frac{x}{t}=const.$, along which we determine the asymptotic behavior of the solution.
In turn, the phase shift (i.e., the Dirichlet spectrum) in this elliptic function depends also on the scattering data, and is computed explicitly via the Jacobi inversion problem.
\end{abstract}

\maketitle
\section{Introduction}

The Korteweg--de Vries (KdV) shock problem is concerned with the long-time behavior of the solution of the KdV equation
\[
q_t(x,t)=6q(x,t)q_x(x,t)-q_{xxx}(x,t), \quad (x,t)\in\R\times\R,
\]
with steplike initial profile 
\beq \label{ini}
\left\{ \begin{array}{ll} q(x,0)\to 0,& \ \ \mbox{as}\ \ x\to +\infty,\\
q(x,0)\to -c^2,&\ \ \mbox{as}\ \ x\to -\infty.\end{array}\right.
\eeq
This behavior is well understood on a physical level of rigor for the pure step initial data (i.e.\ when $q(x,0)=0$ as $x>0$, $q(x,0)=-c^2$ as $x<0$), and
was studied using a quasi-classical Whitham approach in \cite{FW}, \cite{gp1}, \cite{gp2} with further treatments using the matched-asymptotic method
in \cite{AB}, \cite{LN}. Extensions to steplike finite-gap backgrounds were given in \cite{Bik1}--\cite{Bik2} and \cite{N}. 
This led to the following three main regions with different asymptotical behavior of the solution:
\begin{enumerate} [1.]
\item Region $x<-6 c^2 t$, where the solution is asymptotically close to the background $-c^2$ up to a decaying dispersive tail.
\item Middle region $-6c^2 t<x<4c^2 t$, also known as dispersive shock or elliptic region, where the solution can asymptotically be described by a modulated elliptic wave.
\item Soliton region $4c^2 t < x$, where the solution is asymptotically given by a sum of solitons.
\end{enumerate}
We refer to our paper \cite{EGKT} for further details and more on the history of this problem. In this note we want to revisit the middle region which is the most interesting and challenging one from a mathematical point of view.
In particular, in the case of pure step initial data Gurevich and Pitaevskii (\cite{gp1}, \cite{gp2}, see also \cite{LN}) derived the following large-time asymptotical formula in the elliptic region:  
\beq\label{dnu}
q(x,t)\sim q_\text{GP}(x,t)= -2c^2 \mathrm{dn}^2\left(2 t c(6\xi - c^ 2(1 +\mathfrak m^2(\xi)) ,\,\mathfrak m(\xi)\right) +c^2(1 -\mathfrak m^2(\xi)),
\eeq
where $\xi=\frac {x}{12 t},$ and the modulus $\mathfrak m(\xi)$ is determined implicitly by
\beq\label{neav}
\frac{c^2}{6} \left(1 +\mathfrak m^2(\xi)-\frac{2\mathfrak m^2(\xi)(1-\mathfrak m^2(\xi))K(\mathfrak m(\xi))}{[E(\mathfrak m(\xi)) - (1-\mathfrak m^2(\xi))K(\mathfrak m(\xi))]}\right)
= \xi.
\eeq
Here $\mathrm{dn}(s,\mathfrak m)$ is the Jacobi elliptic function and $K(\mathfrak m)$, $E(\mathfrak m)$ are the standard complete elliptic integrals.
Function $q_\text{GP}(x,t)$ is a stationary running wave of the KdV equation if the parameter $\xi$ is a constant. 

On the other hand, in \cite{EGKT} under the assumption
\beq\label{decay}
\int_0^{+\infty} \E^{C_0 x}(|q(x)| + |q(-x)+c^2|dx<\infty,\ \ \ C_0>c>0,
\eeq
 we derived the formula (for precise conditions see our paper)
\beq \label{secon}
q(x,t)\sim q_\text{RH}(x,t) = \frac{\Gamma(\xi)B'(\xi)}{6\pi} \frac{d^2}{dv^2}
 \log\theta_3\left(\frac{t B(\xi) +\Delta(\xi)}{2\pi}+v \right)\Big|_{v=0}- \frac{z'(\xi)}{6},
\eeq
where all quantities are associated with the elliptic Riemann surface $y^2=\la(\la +c^2 )(\la+a^2(\xi))$ depending on the parameter $a(\xi)>0$ for $\xi\in (-c^2/2, c^2/3)$
implicitly determined by 
\beq\label{norm5}
\int_{a(\xi)}^{0}\left(\xi +\frac{c^2 -a^2(\xi)}{2} - s^2\right)\sqrt{\frac{s^2 - a^2(\xi)}{c^2 - s^2}}ds=0.
\eeq
As is shown in \cite{KM}, \cite{KM1}, the point $a(\xi)$ increases monotonically from $0$ to $c>0$ as $\xi$ changes from $-c^2/2$ to $c^2/3$. The quantities
$B$, $\Gamma$, $z$ and $\Delta$ are defined in \eqref{defB}--\eqref{defDelta} below.

Formula \eqref{secon} is very reminiscent of the usual Its--Matveev formula for the one-gap solution of the KdV equation. Hence this raises two natural questions, namely
whether \eqref{secon} agrees with \eqref{dnu} and whether they reduce to the Its--Matveev formula for constant $\xi$, that is, whether \eqref{secon} can be viewed as a slowly
modulated one-gap solution of the KdV equation. The purpose of the present note is to give a positive answer to both questions.

\section{Comparison between $q_\text{RH}(x,t)$ and the Its--Matveev formula}

We begin by recalling the well-known Its--Matveev formula for finite-gap solution of the KdV equation in the case of one gap in the spectrum (see, for example,  \cite{gh1}, \cite{Kuksin}, \cite{M}, \cite{IM}). To facilitate further comparison of this formula with $q_\text{RH}(x,t)$ we suppose that the spectrum of this one gap solution is the set 
$\sigma=[-c^2, -a^2(\xi)]\cup [0,\infty)$. Let $\M=\M(\xi)$ be the elliptic Riemann surface associated with the function
\[
\mathcal R(\la):=\mathcal R(\la,\xi)=\sqrt{\la(\la +c^2 )(\la+a^2(\xi))},
\]
where the cuts are taken along the spectrum $\sigma$, and $\mathcal{R}(\la)$ takes positive values on the upper side of the cut along the interval $[0,\infty)$. A point on $\M$ is denoted by $p=(\la, \pm)$, $\la \in \C$, and
the projection onto $\C \cup \{\infty\}$ is denoted by $\pi(p)=\la$.
The sheet exchange map is given by
$
p^*=
(\la,\mp)  \text{ for } p=(\la,\pm).$
The sets
\begin{align*}
\Pi_U = \{(\la,+) \mid \la \in \C \setminus \sigma\} \subset \M, \quad
\Pi_L = \{p^* \mid p \in \Pi_U\},
\end{align*}
are called the upper, lower sheet, respectively.  Introduce a canonical basis of $\mathfrak a$ and $\mathfrak b$ cycles. The cycle $\mathfrak b$ surrounds the interval $[-c^2, -a^2]$ clockwise on the upper sheet, and the cycle $\mathfrak a$ supplements $\mathfrak b$ passing along the gap from $-a^2$ to $0$ in the positive direction on the lower sheet and back from $0$ to $-a^2$ on the upper sheet. Next, let $d\omega$ be a holomorphic differential on $\M$, normalized by $\int_{\mathfrak a}d\omega=2\pi\I$. Evidently,
\beq\label{omm}
d\omega=\tilde\Gamma \frac{ d\la}{\mathcal R(\la)},\quad \tilde\Gamma=2\pi\I \left(\int_\mathfrak{a} \frac{d\la}{\mathcal R(\la)}\right)^{-1}>0,\eeq
and 
\beq\label{taup} \tilde\tau=\int_\mathfrak{b} d\omega<0.\eeq
Define now the theta function of $\M$ by 
\[
\theta(z)=\theta(z\mid \tilde\tau)=\sum_{m\in\Z}\exp\{\frac{1}{2}\tilde\tau m^2 + z m\}.
\]
Recall that this function is even and takes real values for $z\in\I\R\cup\R$. 

Next, following \cite{Kuksin}, introduce two meromorphic differentials of the second kind $d\Omega_1$ and $d\Omega_3$ with vanishing $\mathfrak a$-periods and with the only pole at infinity of the form:
\beq\label{omm4}
d\Omega_1=\frac{\I}{2}\frac{\la - h}{\mathcal R(\la)}d\la,\quad d\Omega_3=-\frac{3\I}{2}\frac{(\la - \nu_1)(\la - \nu_2)}{\mathcal R(\la)}\,d\la,
\eeq
where
\beq\label{defh}
h=\int_\mathfrak{a} \frac{\la d\la}{\mathcal R(\la)}\left(\int_\mathfrak{a} \frac{ d\la}{\mathcal R(\la)}\right)^{-1} \in (-a^2, 0),
\eeq
and the points $\nu_j\in\R$ are chosen such that $\int_\mathfrak{a}d\Omega_3=0$,  and
\beq\label{nu2}
2\nu_1 + 2\nu_2 + c^2 +a^2=0.
\eeq
The last equality guarantees the absence of the term of order $\la^{-1/2}d\la $ in representation for $d\Omega_3$.
Note, that at least one of the points $\nu_i$ lies in the gap $(-a^2, 0)$. Denote
\beq\label{VW}
\I V=\int_\mathfrak{b} d\Omega_1, \ \mbox{and}\ \ \I W=\int_\mathfrak{b} d\Omega_3.
\eeq
The values $V$ and $W$ are called the wave number and  the frequency. Evidently, $V,W\in\R$. 

Next recall the Abel map $A(p)=\int_\infty^p d\omega$. Let $p_0$ be a point on $\M$ with the projection in the gap, $\pi(p_0)\in[-a^2,0]$. For  such  a point the value $A(p_0) +\mathcal K$ is pure imaginary (cf.\ \cite{dubr}), where $\mathcal K=-\frac{\tilde\tau}{2} +\pi \I$ is the Riemann constant. Then the Its--Matveev formula for the one-gap solution with initial Dirichlet divisor $p_0$ reads (cf.\ \cite{Kuksin}, \cite{M}):
\beq\label{IM}
q_\text{IM}(x,t)=-2\frac{d^2}{dx^2}\log\theta\left(\I V x- 4\I W t - A(p_0) - \mathcal K\right) -a^2 - c^2 - 2 h,
\eeq
where $h$ is defined by \eqref{defh}. 

Now we will study in more details formula \eqref{secon}, where the quantities are defined as follows: $\tau\in\I\R_+$ is the period given by
\beq\label{taupe}
\tau:=\tau(\xi)=-\int_{-c^2}^{-a^2(\xi)} \frac{d\la}{\sqrt{\la(\la+c^2)(\la + a^2(\xi))}}\left( \int_{-a^2(\xi)}^{0} \frac{d\la}{\sqrt{\la(\la+c^2)(\la + a^2(\xi))}} \right)^{-1},
\eeq
and
\[
\theta_3(v)=\theta_3(v\mid\tau)=\sum_{m\in\Z}\exp\{(m^2\tau + 2mv)\pi\I\}
\]
is the associated theta function. Furthermore,
\beq\label{defB}
B(\xi) = 24 \int_{a(\xi)}^{c}\left(\xi +\frac{c^2 -a^2(\xi)}{2} - s^2\right)\sqrt{\frac{s^2 - a^2(\xi)}{c^2 - s^2}}ds;\eeq  \beq\label{Gamdef}
\Gamma(\xi)=
-\frac{1}{2}\left(\int_{-a(\xi)}^{a(\xi)}
\left((c^2 - s^2)(a^2(\xi) - s^2)\right)^{-1/2}ds\right)^{-1};
\eeq
\beq \label{defz}
z(\xi) =\frac{12\xi(c^2 - a(\xi)^2) + 3c^4 + 9a(\xi)^4-6a(\xi)^2c^2}{2};
\eeq
\beq\label{defDelta}
\Delta(\xi) =2\int_{a(\xi)}^c\frac{\log|(\ol T(\I s)T_1(\I s)|}{\sqrt{(c^2 - s^2)(s^2 -a(\xi)^2 )}} ds\left(\int_{-a(\xi)}^{a(\xi)}\frac{ds}{\sqrt{(c^2 - s^2)(a^2(\xi)-s^2)}}\right)^{-1};
\eeq
$T$ and $T_1$ are the left and the right transmission coefficients of the initial data \eqref{ini}.
In formulas \eqref{defB}, \eqref{Gamdef}, and \eqref{defDelta} the positive value of square root is taken.

Denote $-a^2(\xi)=\gamma(\xi):=\gamma$ ,
\beq\label{defmu}
\mu =-\xi -\frac{c^2 +\gamma}{2},
\eeq
and put $\la=-s^2$ in \eqref{norm5}. Then \eqref{norm5} is equivalent to
\beq\label{norm6}
\int_\gamma^0 G(\la) d\la=0, \quad G(\la,\xi):=\frac{(\la - \mu)(\la-\gamma)}{\mathcal R(\la)},
\eeq
or
\beq\label{norm7}
F(\gamma,\xi):=\int_\gamma^0 \left(\la + \xi +\frac{\gamma + c^2}{2}\right)\sqrt{\frac{\la - \gamma}{\la(\la + c^2)}}\ d\la=0.
\eeq
Equation \eqref{norm7} determines the function $\gamma(\xi)$ implicitly, and the function $\mu(\xi)$ by \eqref{defmu}.

\begin{lemma}\label{edecr}
For $\xi\in (-c^2/2,\,c^2/3)$ the function $\gamma(\xi)$ decays monotonically from $0$ to $-c^2$, and  $\mu(\xi)\in (\gamma(\xi), 0)$. Moreover,
\beq\label{deriv}
\frac{d}{d\xi}\gamma(\xi)=4\frac{h(\xi) - \gamma(\xi)}{3\gamma(\xi) + 2\xi +c^2},
\eeq
where $h(\xi)$ is defined by formula \eqref{defh} with $\mathcal R(\la)=\sqrt{\la(\la-\gamma(\xi))(\la + c^2)}$.
\end{lemma}

\begin{proof}
The existence and uniqueness of $\gamma(\xi)$ is proved in \cite{KM}.
Formally differentiating \eqref{norm7} with respect to $\xi$ gives
$\frac{d\gamma}{d\xi}=-\frac{\partial F}{\partial \xi}(\frac{\partial F}{\partial \gamma})^{-1}$. In turn, this implies 
\[
\frac{d\gamma}{d\xi}=\frac{4\mathcal S}{3\gamma + 2\xi +c^2},\quad 
\mathcal S = \int_\gamma^0\frac{\lambda-\gamma}{\mathcal R(\lambda)}d\lambda \left(\int_\gamma^0\frac{d\lambda}{\mathcal R(\lambda)}\right)^{-1},
\]
which implies \eqref{deriv}. 
Evidently $\mathcal S>0$ for $\gamma<0$. Therefore, to prove monotonicity of $\gamma(\xi)$ it is sufficient to prove that 
\[
f(\xi)=3\gamma(\xi)+2\xi+c^2
\]
is negative for $\xi\in (-c^2/2, c^2/3)$, and that $\gamma(\xi)$ is also negative there.  We observe that $F(-c^2, c^2/3)=0$, therefore $\gamma(c^2/3)=-c^2$. Since $f(c^2/3)=-4c^2/3<0$, then $\gamma^\prime(c^2/3)<0$. Thus $\gamma(\xi)$ grows continuously starting from $-c^2$ when $\xi$ decays starting from $c^2/3$. The function $f$ is a continuous function of $\xi$, and the monotonicity of $\gamma$ could only stop if there is a change of sign of $f$. Let
$\xi_0$ be a point where $f(\xi_0)=0$.
Then $\xi_0$ satisfies the system
\[
\begin{cases}
&3\gamma(\xi_0)+2\xi_0+c^2=0\\
&2\mu(\xi_0)-2\xi_0-c^2-\gamma(\xi_0)=0.
\end{cases}
\]
 Thus $\mu(\xi_0)=\gamma(\xi_0)$. But for $\gamma<0$ formula \eqref{norm6} holds iff $\mu\in(\gamma, 0)$. This means that $\mu(\xi_0)=\gamma(\xi_0)=0$, that is $\xi_0=-c^2/2$. 
\end{proof}

\begin{lemma}\label{abelint}
Let $\gamma(\xi)$ and $\mu(\xi)$ be as in Lemma \ref{edecr} and $h(\xi)$ as in \eqref{defh}. Let $G(\la,\xi)$ be defined by  \eqref{norm6} and  $d\Omega_j$, $j=1,3$  by \eqref{omm4}-\eqref{nu2}.  Then
 for any $\xi\in \left(-\frac{c^2}{2}, \frac{c^2}{3}\right)$ the following representation is valid
 \beq\label{cond37}
G\, d\la = \frac{2\I}{3}d\Omega_3 -2\I\xi\,d\Omega_1.
\eeq
 Moreover, the following formula holds
\beq\label{deri}
\frac{\pa}{\pa\xi} G(\la,\xi)= \frac{\la - h(\xi)}{\mathcal R(\la,\xi)}.
\eeq
\end{lemma}

\begin{proof} 
By \eqref{norm6} the differential $G\, d\la$ has vanishing $\mathfrak a$-period. Moreover, using \eqref{nu2}, \eqref{defmu} one checks that $G\, d\la - \frac{2\I}{3}d\Omega_3 +2\I\xi\,d\Omega_1$
has no pole at $\infty$ and hence must vanish. This proves \eqref{cond37}.

To get \eqref{deri} we evaluate
\begin{align}\nn
\frac{\pa}{\pa\xi}G(\la,\xi) - \frac{\la - h(\xi)}{\mathcal R(\la,\xi)}
&=\frac{\pa}{\pa\xi}\left(\frac{(\lambda-\mu)\sqrt{\lambda-\gamma}}{\sqrt{\lambda(\lambda+c^2)}}\right) -\frac{\la - h}{\mathcal R(\la)} \\ \nn
 &=\frac{(-2\mu^\prime - \gamma^\prime)\lambda+2\mu^\prime\gamma+\mu \gamma^\prime -2\la +2h}{2\mathcal R(\lambda)}.
\end{align}
Formula  \eqref{defmu} implies:
\beq\label{muprime}
2\mu^\prime+\gamma^\prime+2=0.
\eeq
Thus, to justify \eqref{deri} one has to prove the equality
\[
2\mu^\prime\gamma+\mu \gamma^\prime  + 2h=0.
\]
 By virtue of \eqref{deriv}, \eqref{muprime}, and \eqref{defmu} we get 
\begin{align*}
2\mu^\prime\gamma+\mu \gamma^\prime  + 2h &= (-\gamma^\prime-2)\gamma -\gamma^\prime\left(\xi +\frac{c^2 + \gamma}{2}\right)+2h \\ 
&= -\frac{1}{2}\gamma^\prime(3\gamma+c^2+2\xi)-2\gamma+2h=0,
\end{align*} 
which proves \eqref{deri}.
\end{proof}

\begin{lemma}
Let $B(\xi)$, $\Gamma(\xi)$,  and $z(\xi)$ be defined by \eqref{defB}, \eqref{Gamdef} and \eqref{defz}, respectively.  Let $V(\xi)$ and $W(\xi)$ be the wave number and the frequency, defined by \eqref{VW} and let $h(\xi)$ be as in \eqref{defh}. Then the following identities hold:
\beq \label{main13} t B(\xi)= -4W(\xi) t  + V(\xi)x,
\eeq
\beq\label{main15} (a)\ \frac{d}{d\xi} B(\xi)= 12 V(\xi), \quad (b)\ 4\pi\Gamma(\xi)=-V(\xi),
\eeq
\beq\label{main17}\frac{1}{6}\frac{d}{d\xi} z(\xi)=  c^2  + a^2 (\xi) + 2 h(\xi).
\eeq
\end{lemma}

\begin{proof}
To get \eqref{main13} we make change of variables $\la=-s^2$  in \eqref{defB}, and take into account  \eqref{cond37}, \eqref{defmu}, $\xi=\frac{x}{12 t}$, 
and the definition of the $\mathfrak b$ period on the Riemann surface $\M(\xi)$. Then
\begin{align*}
t B(\xi)& =-12 t \int_{-a^2}^ {-c^2} (\la - \mu)\sqrt{\frac{-a^2 - \la}{-\la(\la + c^2)}}d\la=12 t \int_{-c^2}^{\gamma} G(\la,\xi)d\la\\ & =4\I t \int_{\mathfrak b} d \Omega_3  - \I x \int_{\mathfrak b} d\Omega_1= -4tW(\xi) + x V(\xi).
\end{align*}
This proves \eqref{main13}. Formula \eqref{main15}, (a) follows from \eqref{cond37} and \eqref{deri}:
\[
\frac{d}{d\xi}B(\xi)=12\frac{d}{d\xi}\int_{-c^2}^{\gamma} G(\la,\xi)d\la=12\gamma^\prime(\xi) G(\gamma(\xi),\xi)-12\I\int_{\mathfrak b} d\Omega_1=12 V(\xi).
\]
Next, by definition of the $\mathfrak a$ period, formula  \eqref{Gamdef} reads:
\beq\label{valG}
\Gamma(\xi)=\left(2\int_0^{-a^2}\frac{d\la}{\sqrt{-\la(\la +c^2)(\la +a^2)}}\right)^{-1}=\frac{\I}{\int_{\mathfrak a}\frac{d\la}{\mathcal R(\la)}}=\frac{\tilde\Gamma}{2\pi},
\eeq
where $\tilde\Gamma$ is the normalization constant from \eqref{omm}.
On the other side, formulas \eqref{omm}, \eqref{omm4}, and the residue theorem (\cite{FK}) yield
\[
\int_{\mathfrak a} d\omega \int_{\mathfrak b}d\Omega_1=2\pi\I \int_{\mathfrak b}d\Omega_1=2\pi\I \res_{\infty} \left(d\Omega_1(p)\int_{\infty}^p d\omega \right),\quad p=(\la,+).
\]
Since in the local parameter $z=\la^{-1/2}$, $z\to 0$, we have
\[
d\Omega_1=\left(\frac{\I}{2\sqrt{\la}} + O\left(\frac{1}{\la^{3/2}}\right)\right)d\la=\frac{-\I}{z^2}(1+o(1))dz,
\]
and 
\[
\int^{\la}_{\infty}d\omega=-\frac{2\tilde\Gamma}{\sqrt\la} +  O\left(\frac{1}{\la^{3/2}}\right)=-2\tilde\Gamma z (1 +o(1)),
\]
then by \eqref{valG}
\[
\res_{\infty} \left(d\Omega_1(p)\int_{\infty}^p d\omega \right)=-2\I\tilde\Gamma= - 4\pi\I \Gamma.
\]
Together with \eqref{VW} this proves \eqref{main15}, (b).
To prove the remaining formula \eqref{main17}, we represent $z(\xi)$ via $\gamma(\xi)$ and apply \eqref{deriv}. Then we get
\begin{align*}
\frac{d}{d\xi}z(\xi) &=\frac{1}{2}\frac{d}{d\xi}\left(12\xi(c^2 +\gamma(\xi)) + 3c^4 + 9\gamma(\xi)^2 + 6\gamma(\xi)c^2\right)\\ &=
6(c^2 +\gamma(\xi)) + \left(6\xi +9\gamma(\xi) +3c^2\right)\gamma^\prime(\xi)\\
&= 6(c^2 +\gamma(\xi)) +12 (h(\xi) - \gamma(\xi))=6(c^2 +a^2(\xi)) + 12 h(\xi).\qedhere
\end{align*}
\end{proof}

Now we are ready to compare formula \eqref{secon} with \eqref{IM}. Comparing the periods $\tau$ defined by \eqref{taupe}, and $\tilde\tau$, defined by \eqref{omm}--\eqref{taup}, we observe that $2\pi\I\tau=\tilde\tau$, and therefore
\[
\theta(z\mid\tilde\tau)=\theta_3(\frac{z}{2\pi\I}\mid\tau).
\] 
Put 
\beq\label{mathcal K}
\mathcal B(\xi)=\frac{\Gamma(\xi)}{6\pi}\left(\frac{d}{d\xi}B(\xi)\right),\quad \mathcal C(\xi)=\frac{1}{6}\frac{d}{d\xi}z(\xi).
\eeq 
Then $q_\text{RH}(x,t)$ can be represented as
 \begin{align} \nn q_\text{RH}(x,t)&=\mathcal B(\xi)\frac{d^2}{dv^2}
 \log\theta_3\left(\frac{\I t B(\xi) +\I \Delta(\xi)+2\I\pi v}{2\pi\I} \right)\mid_{v=0}
+ \mathcal C(\xi), \\ \label{qurh}
&= -4\pi^2 \mathcal B(\xi) \frac{d^2}{dv^2}
 \log\theta\left(\I t B(\xi) +\I \Delta(\xi)+ v \right)\mid_{v=0}
+ \mathcal C(\xi).
\end{align}
On the other hand, by \eqref{main17} we have
\begin{align} \nn
&q_\text{IM}(x,t) -\mathcal C(\xi)\\ \nn &=-2(\I V(\xi))^2 \frac{d^2}{dv^2}\log\theta\left(\I V(\xi) x- 4\I W(\xi) t - A(p_0,\xi) - \mathcal K(\xi) + v\right)\mid_{v=0}.
\end{align} 
Substituting \eqref{main13} and \eqref{main15} into \eqref{mathcal K} and then into \eqref{qurh} we get $\mathcal B(\xi)= -V^2(\xi)$ and $\I t B(\xi)=\I V(\xi)x - 4\I W(\xi) t$.
We conclude that $q_\text{RH}(x,t)=q_\text{IM}(x,t)$ iff there exist a point $p_0$, $\pi(p_0)\in[-a^2, 0]$ such that the following equality is fulfilled:
\[
- A(p_0,\xi) -\mathcal K(\xi)=\I\Delta(\xi) \pmod{2\pi\I},
\] 
where $\Delta(\xi)$ is defined by \eqref{defDelta}. Since
\[
A(-c^2)=\pi\I \pmod{2\pi\I},\ \ A(-a^2)- A(-c^2)=\frac{\tilde\tau}{2},\ \ \mathcal K=-\frac{\tilde\tau}{2} + \pi\I,
\]
and $\Delta(\xi)$  is a real value, then the  point $p_0(\xi)$ can be found as the unique solution of the Jacobi inversion problem (cf.\ \cite{FK}):
\beq\label{jac}
\int_{-a^2(\xi)}^{p_0(\xi)} d\omega=-\I\Delta(\xi).
\eeq
In summary we have proved 
\begin{theorem} 
 For any  fixed $\xi=\frac{x}{12 t}\in (-c^2/2, c^2/3)$ the function $q_\text{RH}(x,t)$ is the usual one-gap solution of the KdV equation:
 \begin{align}\nn
q_\text{RH}(x,t)&=-2\frac{d^2}{dx^2}\log\theta\left(\I V(\xi) x- 4\I W(\xi) t - A(p_0,\xi) - \mathcal K(\xi)\right)\\ \nn
& -a(\xi)^2 - c^2 - 2 h(\xi),
\end{align}
associated with the spectrum  $[-c^2, -a^2(\xi)]\cup[0,\infty)$ and the Dirichlet divisor $p_0(\xi)$, defined via the Jacobi inversion problem \eqref{jac}.
\end{theorem}

Note that since $\theta(z+2\pi\I)=\theta(z)$ we see that $q_\text{RH}(x,t)$ is a periodic function with respect to $x$ and $t$ of periods $\frac{2\pi}{V(\xi)}$ and $\frac{\pi}{2 W(\xi)}$, respectively.

\section{Copmarison between $q_\text{RH}(x,t)$ and the Gurevich--Pitaevskii formula} 

Our next aim is to find a connection between the two functions $\mathfrak m(\xi)$ and $a(\xi)$ implicitly given by \eqref{neav} and \eqref{norm5}.
\begin{lemma} \label{lemperiod}
The following  is true:
\[
\mathfrak m^2(\xi)=\frac{a^2(\xi)}{c^2}.
\]
\end{lemma}

\begin{proof}
Represent \eqref{norm5} as
\begin{align}\label{ellipt}
0&=\left(-\frac{c^2 +a^2(\xi)}{2}+\xi\right)\int_0^{a(\xi)}\sqrt{\frac{ a^2(\xi)-s^2}{c^2 - s^2}}ds \\ \nn
+&\int_0^{a(\xi)}\sqrt{ (a^2(\xi)-s^2)(c^2 - s^2)} ds=\left(\xi -\frac{c^2 +a^2(\xi)}{2}\right) I_1(\xi) + I_2(\xi).
\end{align}
Put $ m:=m(\xi)=\frac{a(\xi)}{c}$. Then (cf.\ \cite{Dwight}, formulas 781.61 and 781.22)
\[
I_1=c (E(m) -  (1-m^2)K(m)),\quad I_2=\frac{c^3}{3} \{(m^2 - 1)K(m) + (m^2 +1)E(m)\},
\]
where $K(m)$ and $E(m)$ are the standard complete elliptic integrals for $0<m<1$. Substituting this into \eqref{ellipt} we obtain:
\begin{align}\nn
\xi&=c^2\left\{ \frac{m^2 +1}{2} - \frac{1}{3}\frac{(m^2 - 1)K(m) + (m^2 +1)E(m)}{E(m) -  (1-m^2)K(m)}\right\}\\ \nn
&=\frac{c^2}{6}\frac{ E(m)(m^2 +1) -(1-m^2)(m^2 +1)K(m) + 2m^2 (m^2-1)K(m)}{E(m) -  (1-m^2)K(m)}\\ \nn
&= \frac{c^2}{6}\left((m^2 +1) +\frac{ 2m^2 (m^2-1)K(m)}{E(m) -  (1-m^2)K(m)}\right).
\end{align}
Thus, \eqref{neav} and \eqref{norm5} define the same function: $\mathfrak m(\xi)=m(\xi)$.
\end{proof}

Lemma \ref{lemperiod} implies that the period \eqref{taupe} corresponds in the standard way (cf.\ \cite{Akh}) to the elliptic modulus $m$ and the following formula is valid:
\[
\mathrm{dn}^2(s,m)=\frac{d^2}{du^2}\log\Theta(u\mid\tau) +\frac{E(m)}{K(m)},
\]
where
\[
\Theta(u\mid\tau)=\theta_3\left(\frac{u}{2 K(m)}+\frac{1}{2}\mid\tau\right)=
\theta_3\left(\frac{u}{2 K(m)}+\frac{1}{2}\right),
\]
and $\theta_3(s)$ is as in \eqref{secon}. 

For the remainder we will fix $\xi$ and omit it from our notation. Fixing $\xi$ we also fix $m=a\,c^{-1}$ and consequently we will also omit $m$ in the complete elliptic integrals: $K:=K(m)$, $E:=E(m)$. Therefore,
\begin{align}\label{8}
q_{GP}(x,t)&=-\frac{c^2}{2K^2}\frac{d^2}{ dv^2}\log\theta_3\left(\frac{ t c(6\xi - c^ 2(1 + m^2))}{K} +\frac{1}{2} +v\right) \\ \nn &+ c^2(1-m^2) -\frac{2E}{K}.
\end{align}
Thus, to convince ourselves that this formula coincides with \eqref{secon} (possibly, up to a phase shift), it is sufficient to show:

\begin{lemma} The following equalities hold:
\begin{align}\label{1}-\frac{c^2}{2K^2}&=\frac{\Gamma}{6\pi}\frac{d B}{d\xi};\\
\label{2} -\frac{6 c \xi - c^3 (1+m^2)}{K} &=\frac{B}{2\pi};\\
\label{3} c^2   -\frac{E}{K}&=  - h,\end{align}
where $B$, $\Gamma$, and $h$ are defined by \eqref{defB}, \eqref{Gamdef}, and \eqref{defh}, respectively.
\end{lemma}

\begin{proof}
Formulas \eqref{main15},  \eqref{Gamdef}, and 781.01 of \cite{Dwight} imply:
\beq\label{new1}
\frac{\Gamma}{6\pi}\frac{d B}{d\xi}=-8\Gamma^2=-\frac{1}{2}\left(\int_0^a\frac{ds}{\sqrt{(c^2 - s^2)(a^2 - s^2)}}\right)^{-2}=-\frac{c^2}{2K^2}.
\eeq
Next, by \eqref{defh} and formula 781.11 from \cite{Dwight}:
\[
- h=\int_0^a\frac{s^2 ds}{\sqrt{(c^2 - s^2)(a^2 - s^2)}}\left(\int_0^a\frac{ ds}{\sqrt{(c^2 - s^2)(a^2 - s^2)}}\right)^{-1}=\frac{c(K -E)}{\frac{1}{c}K}.
\]
This gives \eqref{3}. 
To prove \eqref{2} we use \eqref{main13} and \eqref{main15}. Namely, we have
\[
\frac{t B}{2\pi}=-2\Gamma x -4 Wt =-2 (x -2(c^2 +a^2)t) \Gamma.
\]
Here we used equation (4.4.13) of \cite{M}, which implies the equality $W=(c^2 + a^2) \Gamma$. By \eqref{new1} we have $\Gamma=\frac{c}{4K}$. Thus
\[
\frac{ B}{2\pi}=-\frac{t c}{K} \left(6\xi -(c^2 + a^2)\right),
\]
which proves \eqref{2}. Note that the opposite sign with respect to \eqref{8} of the first summand in $\theta_3$ is not essential as $\theta_3$ is even.
\end{proof}

We proved that formulas \eqref{secon} and \eqref{dnu} represent the same function up to the phase shift. Namely, instead of the summand $(2\pi)^{-1} \Delta$ in the argument of the theta function for $q_\text{RH}(x,t)$, in the same formula for $q_\text{GP}(x,t)$ we have the summand $\frac{1}{2}$. Recall now that the transmission coefficients for the pure step potential have the representation (cf.\ \cite{BF}, \cite{EGLT})
\[
T(k)=\frac{2\I k}{w(k)},\ \  T_1(k)=\frac{2\I \sqrt{k^2 + c^2}}{w(k)},\quad w(k)=\I k + \I\sqrt{k^2 + c^2}.
\]
Since in \eqref{defDelta} $k=\I s$, $s\in [0,c]$, then
\[
|w(k)|^2=|-s +\I\sqrt{c^2 - s^2}|^2=c^2,
\]
and
\[
|T(\I s)T_1(\I s)|=4s c^{-2}\sqrt{c^2 - s^2},\quad s\in [0,c].
\]
Therefore, the value of the phase shift for the pure step case is given by
\begin{align*}
\frac{\Delta_{ps}}{2\pi}&=\frac{\int_a^c\log\left(4sc^{-2}\sqrt{c^2 - s^2} \right)\left((c^2 - s^2)(s^2 - a^2)\right)^{-1/2} ds}{2\pi \int_0^a\left((c^2 - s^2)(s^2 - a^2)\right)^{-1/2}\,ds}\\
&=\frac{1}{2\pi K(m)}\int_m^1\frac{\log\left(4s\sqrt{1 - s^2} \right)}{\sqrt{(1 - s^2)(s^2 - m^2)} }ds .
\end{align*}

\noindent{\bf Acknowledgments.} We thank Alexei Rybkin and Johanna Michor for valuable discussions on this topic. I.E.\ is indebted to the Department of Mathematics at the University of Vienna for its hospitality and support during the fall of 2015, where part of this work was done.

\end{document}